%% file: main.tex
\documentclass[11pt,a4paper]{article}

\usepackage[utf8]{inputenc}
\usepackage[T1]{fontenc}
\usepackage{amsmath,amssymb,amsthm}
\usepackage{algorithm}
\usepackage{algpseudocode}
\usepackage{booktabs}
\usepackage{graphicx}
\usepackage{hyperref}
\usepackage{cleveref}
\usepackage{xcolor}
\usepackage{listings}
\usepackage{tikz}
\usetikzlibrary{arrows.meta,positioning,shapes.geometric,fit}
\usepackage{subcaption}
\usepackage[numbers,sort&compress]{natbib}
\usepackage{geometry}
\newtheorem{theorem}{Theorem}
\newtheorem{definition}{Definition}
\newtheorem{lemma}{Lemma}
\newtheorem{corollary}{Corollary}
\newtheorem{property}{Property}

\newcommand{\cmark}{\textcolor{green!60!black}{\checkmark}}
\newcommand{\xmark}{\textcolor{red!70!black}{$\times$}}

\newcommand{\armname}{ARM}

\title{Causality Laundering: Denial-Feedback Leakage in Tool-Calling LLM Agents}

\author{
  Mohammad Hossein Chinaei\thanks{Corresponding author. Email: \href{mailto:mohammadh.chinaei@gmail.com}{mohammadh.chinaei@gmail.com}. LinkedIn: \url{https://www.linkedin.com/in/mohammadhossein-chinaei/}}\\
  Independent Researcher
}

\date{April 2026}

\begin{document}

\maketitle

\begin{abstract}
\input{sections/00-abstract}
\end{abstract}

\section{Introduction}
\label{sec:introduction}
\input{sections/01-introduction}

\section{Threat Model}
\label{sec:threat-model}
\input{sections/02-threat-model}

\section{Causality Laundering}
\label{sec:causality-laundering}
\input{sections/03-causality-laundering}

\section{ARM Architecture}
\label{sec:architecture}
\input{sections/04-architecture}

\section{Provenance Graph}
\label{sec:provenance-graph}
\input{sections/05-provenance-graph}

\section{Implementation}
\label{sec:implementation}
\input{sections/06-implementation}

\section{Evaluation}
\label{sec:evaluation}
\input{sections/07-evaluation}

\section{Related Work}
\label{sec:related-work}
\input{sections/08-related-work}

\section{Discussion and Future Work}
\label{sec:discussion}
\input{sections/09-discussion}

\section{Conclusion}
\label{sec:conclusion}
\input{sections/10-conclusion}

\bibliographystyle{plainnat}
\bibliography{references}

\appendix
\section{Formal Proofs}
\label{app:proofs}
\input{sections/appendix-proofs}

\end{document}

%% file: sections/00-abstract.tex
Tool-calling LLM agents can read private data, invoke external services, and
trigger real-world actions, creating a security problem at the point of tool
execution.  We identify a denial-feedback leakage pattern, which we term
\emph{causality laundering}, in which an adversary probes a protected action,
learns from the denial outcome, and exfiltrates the inferred information through
a later seemingly benign tool call.  This attack is not captured by flat
provenance tracking alone because the leaked information arises from causal
influence of the denied action, not direct data flow.

We present the \emph{Agentic Reference Monitor} (\armname{}), a runtime
enforcement layer that mediates every tool invocation by consulting a provenance
graph over tool calls, returned data, field-level provenance, and denied actions.
\armname{} propagates trust through an integrity lattice and augments the graph
with counterfactual edges from denied-action nodes, enabling enforcement over both
transitive data dependencies and denial-induced causal influence.  In a controlled
evaluation on three representative attack scenarios, \armname{} blocks causality
laundering, transitive taint propagation, and mixed-provenance field misuse that
a flat provenance baseline misses, while adding sub-millisecond policy evaluation
overhead.  These results suggest that denial-aware causal provenance is a useful
abstraction for securing tool-calling agent systems.

%% file: sections/01-introduction.tex
Tool-calling LLM agents increasingly interact with external systems by invoking
tools that read data, call services, and trigger consequential side
effects~\citep{schick2023toolformer,patil2023gorilla}.  A single agent session may
query a user's inbox, retrieve financial records, draft a document, and send it to
a third party, all mediated by protocols such as MCP~\citep{anthropic2024mcp} that
expose a growing catalog of callable tools.  The critical security boundary in such
systems is not text generation itself, but tool execution: every tool call is an
access decision, and every tool result imports data from a source the agent does
not control.  Securing these systems therefore requires principled enforcement at
the tool boundary, rather than relying solely on post hoc detection or filtering.

Recent work has made real progress on runtime enforcement for tool-calling agents.
Information-flow-control systems such as FIDES~\citep{costa2025fides} enforce
label-based restrictions over observed flows.  Secure-by-design architectures
such as CaMeL~\citep{debenedetti2025defeating} extract trusted control and data
flow from the user query and enforce capability-style restrictions.  Graph-based
approaches such as PCAS~\citep{palumbo2026pcas} construct dependency graphs over
tool-call events and evaluate declarative policies against them.  Causal attribution
defenses~\citep{zhang2026agentsentry,kim2026causalarmor} use counterfactual
re-execution to determine whether a proposed action is driven by user intent or by
injected content.  These systems represent genuine progress; collectively, they show
that principled runtime enforcement for agents is both feasible and necessary.
However, they primarily model successful executions, returned data, and observed
action traces, rather than explicitly representing denied actions as events with
potential downstream causal influence.

A blocked action is not always the end of an attack.  Consider an agent that, under
adversarial influence, attempts to read a user's salary record.  The enforcement
layer denies the call, and no sensitive bytes are returned.  Yet the denial itself
may still reveal security-relevant information: the agent can infer that such a
record exists, that it is protected, or, if the denial includes a reason, that the
record belongs to a particular category of sensitive data.  A later tool call that
mentions ``compensation'' to an external party can therefore leak information that
was never directly read from the protected resource.  The leakage arises from the
denial outcome, not from any successful data transfer.  Flat taint tracking does not
capture this behavior because no sensitive value ever enters the ordinary tool-result
channel, and dependency graphs over successful executions contain no direct data-flow
edge from the blocked action to the later exfiltration step.  Existing causal
attribution checks may also miss this pattern, because the triggering influence is
the denial event itself rather than untrusted content returned by the tool.  Recent
empirical evidence supports the plausibility of this channel: frontier LLM agents
have been shown to infer hidden monitoring purely from blocking feedback, even when
they are not explicitly told that monitoring is
present~\citep{jiralerspong2026noticing}.

We term this pattern \emph{causality laundering}: an adversary probes a protected
action, learns from the denial outcome, and exfiltrates the inferred information
through a later seemingly benign tool call.  Existing provenance models over
successful executions do not explicitly record this dependency, because the
relevant influence arises from the denial event rather than from returned data.
Causality laundering is a
specific instance of a broader class of denial-feedback implicit flows in
tool-calling agents.  To our knowledge, existing defenses do not explicitly
represent denied actions as first-class provenance events and enforce over their
downstream influence.

To address this gap, we present the \emph{Agentic Reference Monitor} (\armname{}),
a provenance-aware runtime enforcement layer that treats denied actions as
first-class events in the execution graph.  When a tool call is denied,
\armname{} records the denial as a node in the provenance graph and introduces
counterfactual edges to subsequent actions that may have been causally influenced
by it.  Trust then propagates through an integrity lattice over transitive data
dependencies, field-level provenance, and denial-induced counterfactual paths,
allowing downstream actions shaped by a denial to inherit the security-relevant
context of the denied event.  All policy decisions are computed by deterministic
graph traversals and explicit rules, rather than delegated back to the LLM under
scrutiny.

\paragraph{Contributions.}  Our contributions are threefold:

\begin{enumerate}
  \item We identify and formalize \emph{causality laundering}, a denial-feedback
    leakage pattern in which an adversary probes a protected action, infers
    security-relevant information from the denial outcome, and exfiltrates it
    through a subsequent tool call (\Cref{sec:causality-laundering}).

  \item We present \armname{}, a provenance-aware runtime enforcement layer that
    explicitly represents denied actions, counterfactual dependencies, transitive
    data dependencies, and field-level provenance, and enforces policy through
    deterministic graph traversal over an integrity lattice
    (\Cref{sec:architecture,sec:provenance-graph}).

  \item We evaluate \armname{} on three representative attack scenarios and show
    that the graph-aware engine blocks all three, namely causality laundering,
    transitive taint propagation, and mixed-provenance field misuse, while a flat
    provenance baseline misses all three, with sub-millisecond evaluation overhead
    (\Cref{sec:evaluation}).
\end{enumerate}

%% file: sections/02-threat-model.tex
\subsection{System and Trust Boundary}

We consider a deployment in which a single LLM agent interacts with the external
world exclusively through tool calls.  \armname{} mediates this boundary: every
attempted tool invocation is evaluated before execution, and every tool result is
imported into the agent context only after policy evaluation and provenance
recording.  The primary security-relevant actions in scope are reads from protected
resources, writes or updates to external systems, and outbound communication to
potentially adversarial sinks.

The agent interacts with external systems through a set
of tools exposed via the Model Context Protocol (MCP) or an equivalent
tool-calling interface.  The agent receives a system prompt from a trusted
operator, user messages from a trusted end user, and tool outputs from
external services of varying trustworthiness.

The enforcement point is the tool-call boundary: \armname{} interposes on every
tool invocation before execution and on every tool result before it enters the
agent's context.  The protected assets include confidential records (e.g.,
personnel files, credentials, financial data), and the privileged sinks include
outbound communication (e.g., \texttt{send\_email}, \texttt{post\_message}),
file mutation (e.g., \texttt{write\_file}, \texttt{delete\_file}), external API
calls (e.g., \texttt{transfer\_funds}, \texttt{update\_record}), and code
execution (e.g., \texttt{run\_code}).  A security violation occurs when a privileged sink is reached by an action whose
provenance includes either untrusted data or denial-induced causal influence,
without explicit policy authorization.

\subsection{Attacker Model}

The attacker acts indirectly by controlling content that the agent observes, not
by modifying \armname{} itself.  The attacker's goal is to cause the agent to
execute unauthorized tool calls or exfiltrate sensitive data through tool-call
arguments.  We assume the attacker \emph{cannot} modify the agent's code, the
\armname{} enforcement layer, or the system prompt.  The attacker \emph{can}:

\begin{enumerate}
  \item \textbf{Inject adversarial content into tool outputs.}  A compromised or
    adversary-controlled tool returns data containing prompt injection
    payloads~\citep{greshake2023youve}.  This is the dominant real-world attack
    vector: poisoned emails, malicious web content, and tampered API responses.

  \item \textbf{Craft multi-step attack chains.}  The attacker's payload may not
    exfiltrate data in a single tool call.  Instead, it may cause the agent to
    (a)~read sensitive data, (b)~transform it through intermediate tool calls to
    launder its provenance, and (c)~exfiltrate the result through a seemingly
    benign tool call.

  \item \textbf{Exploit denial signals.}  When \armname{} denies a tool call, the
    agent observes the denial.  A prompt-injected agent may infer information from
    the denial itself (e.g., ``the file exists because the read was denied on
    permissions, not on file-not-found'') and exfiltrate this inference.

  \item \textbf{Manipulate tool descriptions.}  In MCP-based deployments, tool
    descriptions are metadata provided by the tool server and may be modified by
    a compromised or malicious server~\citep{invariantlabs2024mcp}.  A manipulated
    description can influence the agent's tool selection and argument construction.
\end{enumerate}

\subsection{Trust Sources}

The system model above distinguishes trusted inputs (system prompt, user messages)
from tool outputs of varying trustworthiness.  We formalize this distinction as a
total order over data sources:

\begin{definition}[Integrity Lattice]
\label{def:integrity-lattice}
Let $\mathcal{T}$ be the set of five trust levels with total order:
\[
  \textnormal{\textsc{ToolDesc}} <_\mathcal{T}
  \textnormal{\textsc{ToolUntrusted}} <_\mathcal{T}
  \textnormal{\textsc{ToolTrusted}} <_\mathcal{T}
  \textnormal{\textsc{UserInput}} <_\mathcal{T}
  \textnormal{\textsc{SysInstr}}
\]
The trust level of a data item derived from multiple sources is the minimum of
its sources' trust levels (conservative join).
\end{definition}

This lattice reflects operational reality: system instructions are the most
trusted (set by the operator), user input is trusted but less authoritative, and
tool outputs range from trusted (e.g., a signed API response) to untrusted
(e.g., raw web scrape content).  Tool descriptions---metadata provided by MCP
servers---are the least trusted, as they can be manipulated in tool-poisoning
attacks~\citep{invariantlabs2024mcp}.

\subsection{Security Goals}
\label{sec:security-goals}

In practical terms, \armname{} aims to ensure that (1)~no tool call influenced by
untrusted data or denial-induced causal influence reaches a privileged sink
without an explicit policy decision,
(2)~every enforcement decision is recorded in a way that resists after-the-fact
tampering, and (3)~no single enforcement check is a single point of failure.  We
formalize these as three properties:

\begin{enumerate}
  \item \textbf{Mediated Integrity.}  No causal path exists from an untrusted
    source to an \textsc{Allow} verdict on a privileged tool call without
    traversing at least one enforcement check (\Cref{thm:mediated-integrity}).

  \item \textbf{Tamper-Evident Audit.}  Every enforcement decision is recorded in a
    hash-chained audit log.  Tampering with any entry invalidates the chain.

  \item \textbf{Defense in Depth.}  Multiple independent enforcement layers operate
    in series; an attacker must simultaneously evade all layers to succeed.
\end{enumerate}

\subsection{Scope and Non-Goals}

\armname{} enforces \emph{integrity} properties: it prevents untrusted data from
triggering unauthorized actions.  It provides \emph{explicit secrecy} (blocking
direct data flows to unauthorized sinks) but does \emph{not} claim full
non-interference for implicit flows (e.g., timing channels, control-flow
inference).  Full implicit-flow tracking is impractical in LLM-based systems where
the model's internal reasoning is opaque.  However, \armname{}'s counterfactual
edges (\Cref{sec:causality-laundering}) mitigate a significant class of
inference-based attacks that would otherwise be invisible.

We treat malicious-user attacks as out of scope and focus on indirect compromise
through untrusted tool-observed content.  We consider a single-agent setting in
this paper.  Multi-agent extensions (delegation tokens, vector clocks, composition
checking) are discussed in \Cref{sec:discussion} and deferred to future work.

%% file: sections/03-causality-laundering.tex

Runtime enforcement systems for tool-calling agents often adopt a common implicit
assumption: once a tool call is denied, the threat has been neutralized.  The
blocked action does not execute, no sensitive data is returned, and the system
proceeds as though the attack has ended.  We argue that this assumption is
incorrect.  A denied action is not merely a failed operation; it is an observable
event whose outcome can influence the agent's subsequent behavior.  Denial
feedback can therefore carry security-relevant information that is not captured
by provenance models over successful executions alone.  This section introduces
\emph{causality laundering}, a denial-feedback leakage pattern that exploits
this gap, and develops the intuition and formalism needed to detect it.

\subsection{Motivating Example}
\label{sec:cl-example}

Consider an agent with access to \texttt{read\_file} and \texttt{send\_email}
tools.  A prompt injection payload instructs the agent to exfiltrate the contents
of \texttt{/etc/shadow}.  A policy engine correctly denies the
\texttt{read\_file("/etc/shadow")} call because the path matches a sensitive
filesystem pattern.

A flat taint-tracking system considers the attack blocked: no sensitive data
entered the agent's context.  However, the agent can still infer that the file
\emph{exists} if the denial reason is ``permission denied'' rather than ``file
not found.''  The injected payload then instructs the agent to call
\texttt{send\_email(body="shadow file exists on this host")}.  From the
perspective of flat provenance or taint-based enforcement, this second call
appears benign: the payload is short, contains no credential patterns, matches
no sensitive data returned by any tool, and has no data-flow lineage connecting
it to the denied read.  Yet it conveys information derived from the denied
action.  Blocking the first step did not eliminate the attack; it left behind the
information channel through which the second step operates.

\subsection{Why Existing Defenses Miss the Attack}
\label{sec:cl-gap}

Three broad categories of runtime defense fail to detect the attack above, for
structurally related reasons.

\paragraph{Flat taint and information-flow tracking.}
Systems such as FIDES~\citep{costa2025fides} propagate integrity labels through
tool outputs: when a tool call succeeds, its result inherits the trust level of
its inputs.  When a call is denied, no successful output is produced, so no label
is propagated through the ordinary data-flow channel.  The
\texttt{send\_email} call in the motivating example carries no tainted bytes,
and flat taint correctly reports that no protected data flowed through a
successful tool result.  The leak is therefore invisible to this abstraction,
because it does not travel through returned data; it travels through the agent's
observation of the denial outcome.

\paragraph{Dependency graphs over successful executions.}
Graph-based systems such as PCAS~\citep{palumbo2026pcas} construct dependency
graphs over tool-call sequences and evaluate reachability policies against them.
A denied call, however, is not typically represented as a provenance-bearing
dependency node with downstream causal semantics: it produced no successful
output for later nodes to depend on.  As a result, the causal link from the
denied \texttt{read\_file} to the subsequent \texttt{send\_email} is not
captured in a graph that records only successful execution dependencies.

\paragraph{Causal attribution defenses.}
A separate line of work uses causal attribution to detect indirect prompt
injection.  AgentSentry~\citep{zhang2026agentsentry} and
CausalArmor~\citep{kim2026causalarmor} re-execute agent trajectories with and
without suspected inputs to estimate causal influence.  Their ablation targets
are typically attacker-controlled content that appeared in tool
\emph{results}, such as injected text or manipulated API responses.  Denial
feedback, by contrast, is generated by the enforcement layer itself rather than
by an external attacker, and existing attribution pipelines are not designed to
model denied actions as explicit ablation targets with downstream provenance.

\paragraph{The missing abstraction.}
In each case, the denied action disappears from the defense's world model.  The
agent, however, still observes the denial, and its subsequent behavior is
conditioned on that observation.  The missing object is therefore the denial
event as a first-class provenance node: one that participates in downstream
causal chains even though it produced no successful tool output.  \armname{} targets a class of implicit-flow behavior that is not the focus of
FIDES's threat model.  Causality laundering exploits precisely this gap.

\subsection{Formal Definition}
\label{sec:cl-definition}

\begin{definition}[Causality Laundering]
\label{def:causality-laundering}
Let $G = (V, E)$ be the provenance graph of an agent session.  A
\emph{causality laundering attack} occurs when:
\begin{enumerate}
  \item An action $a_d \in V$ is denied by the enforcement layer at time $t_d$.
  \item A subsequent action $a_s \in V$ at time $t_s > t_d$ produces an
    observable side effect (e.g., sends data to an external sink).
  \item No successful data-flow path connects $a_d$ to $a_s$ through returned
    tool outputs.
  \item The occurrence or content of $a_s$ is \emph{causally influenced} by the
    denial of $a_d$; that is, the agent's behavior would differ had $a_d$ not occurred, or had its denial not been observed.
\end{enumerate}
\end{definition}

Clauses~1 and~2 establish the temporal structure: a denied action precedes a
side-effecting action.  Clause~3 is what makes the attack invisible to existing
data-flow defenses: because no successful data-flow path connects the two
actions, systems that track only explicit propagation through tool results will
see no dependency between them.  Clause~4 is the distinguishing condition: the
later action would have been different had the denial not occurred, meaning the
denial itself carried information that shaped subsequent behavior even without a
returned-data path.  Together, clauses~3 and~4 capture the core tension:
the causal influence is real, but the ordinary data-flow evidence is absent.
In \Cref{sec:counterfactual-detection}, \armname{} operationalizes this
definition conservatively by recording denied actions explicitly and linking them
to subsequent actions via counterfactual edges.

\subsection{Relationship to Established Concepts}
\label{sec:cl-established}

Causality laundering is not conceptually isolated; it is the tool-calling-agent
manifestation of several classical ideas in information flow and causality.

\paragraph{Covert channels (Lampson, 1973).}  Lampson's covert channel
taxonomy~\citep{lampson1973note} distinguishes storage channels (shared state)
from timing channels (observable delays).  Causality laundering is a
\emph{storage channel} in which the shared state is the agent's observed
execution history: the denial event is effectively ``stored'' in the agent's
state and later reflected in a subsequent action.

\paragraph{Counterfactual causality (Pearl, 2009).}  In Pearl's structural causal
model framework~\citep{pearl2009causality}, $a_d$ is a cause of $a_s$ if, in the
counterfactual world where $a_d$ did not occur, $a_s$ would not have occurred
(or would have had different content).  Our detection mechanism operationalizes a
conservative approximation of this idea: \armname{} records a
\textsc{Counterfactual} edge from $a_d$ to $a_s$ and queries whether any path
to a sink traverses such an edge.

\paragraph{Implicit information flows (Denning, 1976).}  Denning's lattice model
for information flow~\citep{denning1976lattice} acknowledges that assignments in
conditional branches create implicit flows.  In the agent setting, the LLM's
``conditional branch'' is opaque: we cannot inspect the internal reasoning that
transforms denial feedback into later behavior.  The observable pattern, however, can still be represented and enforced over: a
denied action followed by a suspicious downstream action.

\subsection{Attack Variants}
\label{sec:cl-variants}

We identify three variants of causality laundering, distinguished by how
denial-derived information is inferred, accumulated, and transmitted.

\begin{enumerate}
  \item \textbf{Denial inference exfiltration.}  A single denied probe leaks a
    small but security-relevant fact (e.g., that a file exists), and the agent
    immediately exfiltrates that inference.  This is the motivating example
    above.  The risk is that even one denial can reveal the presence, absence, or
    sensitivity class of a protected resource, which may already be sufficient
    for reconnaissance or targeted follow-up attacks.

  \item \textbf{Multi-probe fingerprinting.}  The attacker issues a sequence of
    probes (e.g., testing multiple file paths or API endpoints), collects a
    vector of allow/deny outcomes, and exfiltrates the aggregate.  Each
    individual denial reveals little in isolation; the sequence reveals much.
    The risk is cumulative: a systematic scan of $n$ resources can leak enough
    denial feedback to fingerprint the system's access-control configuration or
    map protected assets.

  \item \textbf{Laundered composition.}  The attacker inserts an allowed
    intermediate tool between the denial and the final sink-bound action in
    order to obscure the connection.  For example: probe $\rightarrow$ denied
    $\rightarrow$ use a calculator tool to encode the denial status as a number
    $\rightarrow$ include that number in an outbound message.  The risk is that
    the intermediate transformation launders the denial-derived signal into a
    form that breaks simple temporal, textual, or surface-level correlation with
    the original denied action.
\end{enumerate}

\subsection{Detection Intuition}
\label{sec:cl-intuition}

If the problem is that denied actions are absent from the provenance model, the
natural defense is to record them explicitly and propagate their potential
downstream influence.  The challenge is that the path from a denied action to a
later tool call passes through the LLM's internal state, which is opaque to the
enforcement layer.  We therefore cannot directly prove that a particular
subsequent action was caused by a particular denial.  What we \emph{can} observe
is the external pattern: a denial occurred, the agent continued, and later tool
calls were issued under a context that now includes that denial outcome.

\armname{}'s detection strategy accepts this opacity and works with what is
observable.  Rather than attempting to determine \emph{whether} the model's
internal reasoning was influenced, \armname{} records \emph{that} a denial
occurred and conservatively treats subsequent tool calls as potentially
influenced by that event.  This is a deliberate over-approximation: it may flag
benign tool calls that happen to follow a denial, but it makes denial-induced
causal influence explicit in the provenance model instead of leaving it
unrepresented.  The next subsection operationalizes this intuition using
denied-action nodes and counterfactual edges in the provenance graph.

\subsection{Detection via Counterfactual Edges}
\label{sec:counterfactual-detection}

\armname{} detects causality laundering by maintaining \textsc{Counterfactual}
edges in the provenance graph.  When an action is denied, \armname{} records a
\textsc{DeniedAction} node.  The next tool call automatically receives a
\textsc{Counterfactual} edge from that denial node.  This is a temporal
heuristic, chosen because the model's internal reasoning is opaque: we cannot
determine exactly which later actions were influenced by the denial, so we
conservatively mark the temporally adjacent action as potentially influenced.

At evaluation time, \armname{}'s graph-aware provenance layer
(\Cref{sec:provenance-graph}) queries whether the current tool call is reachable
from any denied-action node via a path containing a \textsc{Counterfactual}
edge.  If so, the call is denied with the reason ``causality laundering
detected.''

This mechanism is a conservative approximation of causal influence, not perfect
causal identification.  It may therefore flag benign tool calls that happen to
follow a denial, producing false positives.  We discuss this trade-off and
possible refinements, such as configurable temporal windows and richer
structural conditions, in \Cref{sec:discussion}.  In practice, prompt-injection
exfiltration attempts often follow closely after the denied probe, making this
heuristic effective for the dominant probe-then-exfiltrate pattern.

%% file: sections/04-architecture.tex

\armname{} is an interposition layer that sits between a tool-calling LLM agent
and its upstream tool servers.  Every tool call is mediated by \armname{} before
execution, and every tool result passes back through \armname{} before entering
the agent's context.  The outcome of each mediation is either \textsc{Allow},
meaning the call proceeds, or \textsc{Deny}, meaning the call is blocked and a
denial record is created.  No tool invocation reaches an external system without
first passing through the enforcement pipeline described in this section.
While we describe \armname{} as an MCP proxy for concreteness, the same
mediation pattern applies to equivalent tool-calling runtimes that expose a
centralized execution boundary.

\subsection{Interposition and Execution Flow}
\label{sec:execution-flow}

The end-to-end path for a single tool call proceeds as follows.  The LLM client
issues a \texttt{tool\_use} message.  \armname{}, operating as an MCP proxy
server, receives the message and evaluates it through the layered policy
pipeline (\Cref{sec:pipeline}).  If all layers return \textsc{Pass}, \armname{}
forwards the call to the upstream tool server, receives the result, records the
returned data and its provenance (\Cref{sec:provenance-graph}), and returns the
result to the agent.  If any layer returns \textsc{Deny}, the call is not
forwarded; instead, \armname{} records a \textsc{DeniedAction} node in the
provenance graph and returns an error to the agent.  In both cases, the
enforcement decision and its per-layer results are written to the audit log.

\subsection{Reference Monitor Properties}
\label{sec:refmon-properties}

\armname{} achieves the three classical reference monitor
properties~\citep{anderson1972planning}:

\begin{enumerate}
  \item \textbf{Complete mediation.}  \armname{} operates as the sole MCP proxy
    between the LLM client and the upstream tool server.  Every
    \texttt{tool\_use} message passes through \armname{}'s \texttt{evaluate()}
    function before reaching the real tool.  Under the deployment model in
    \Cref{sec:execution-flow}, there is no alternative path.

  \item \textbf{Tamper-proofness.}  The enforcement layer runs in a separate
    process from the LLM.  The agent cannot modify policy rules, capability
    tokens, or the provenance graph.  Capability tokens are immutable
    (\texttt{frozen} dataclasses), and the audit log is hash-chained.

  \item \textbf{Verifiability.}  The policy pipeline consists of four layers
    with well-defined, deterministic semantics (described below).  Each layer's
    behavior can be analyzed and tested in isolation.
\end{enumerate}

\subsection{Layered Policy Pipeline}
\label{sec:pipeline}

\armname{} evaluates each tool call through an ordered sequence of policy layers.
The layers are arranged from unconditional boundary checks to provenance-aware
and operator-defined policy checks.  Each layer returns one of two verdicts:

\begin{itemize}
  \item \textsc{Pass}: the layer has no objection (evaluation continues to the
    next layer).
  \item \textsc{Deny}: the layer blocks the call (evaluation halts immediately).
\end{itemize}

The composition rule is strict: \textsc{Deny} from any layer terminates
evaluation.  Only if all layers return \textsc{Pass} is the call forwarded to the
upstream tool.

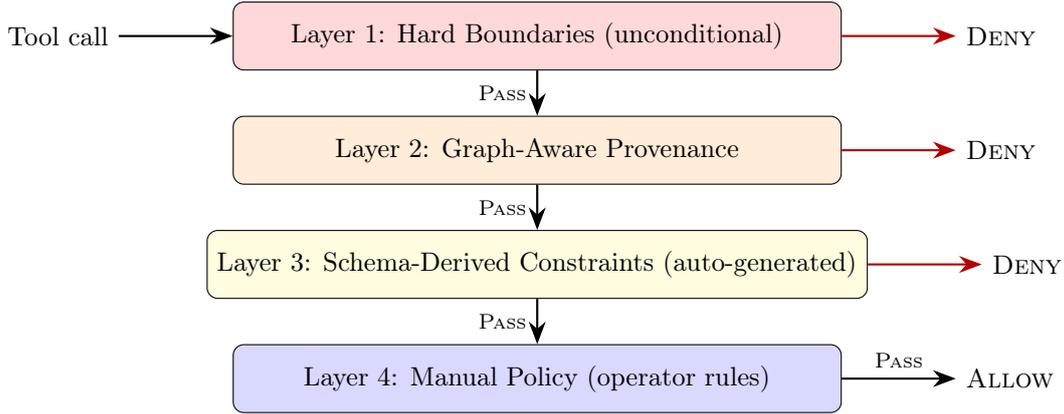
\begin{figure}[t]
  \centering
  \begin{tikzpicture}[
    layer/.style={rectangle, draw, rounded corners, minimum width=8cm,
      minimum height=0.9cm, font=\small},
    arrow/.style={-{Stealth[length=3mm]}, thick},
    node distance=0.6cm
  ]
    \node[layer, fill=red!15] (l1) {Layer 1: Hard Boundaries (unconditional)};
    \node[layer, fill=orange!15, below=of l1] (l2)
      {Layer 2: Graph-Aware Provenance};
    \node[layer, fill=yellow!15, below=of l2] (l3)
      {Layer 3: Schema-Derived Constraints (auto-generated)};
    \node[layer, fill=blue!15, below=of l3] (l4)
      {Layer 4: Manual Policy (operator rules)};

    \node[left=1.5cm of l1, font=\small] (input) {Tool call};
    \node[right=1.5cm of l4, font=\small] (allow) {\textsc{Allow}};
    \node[right=1.5cm of l1, font=\small] (deny1) {\textsc{Deny}};
    \node[right=1.5cm of l2, font=\small] (deny2) {\textsc{Deny}};
    \node[right=1.5cm of l3, font=\small] (deny3) {\textsc{Deny}};

    \draw[arrow] (input) -- (l1);
    \draw[arrow] (l1) -- (l2) node[midway, left, font=\scriptsize] {\textsc{Pass}};
    \draw[arrow] (l2) -- (l3) node[midway, left, font=\scriptsize] {\textsc{Pass}};
    \draw[arrow] (l3) -- (l4) node[midway, left, font=\scriptsize] {\textsc{Pass}};
    \draw[arrow] (l4) -- (allow) node[midway, above, font=\scriptsize] {\textsc{Pass}};
    \draw[arrow, red!70!black] (l1) -- (deny1);
    \draw[arrow, red!70!black] (l2) -- (deny2);
    \draw[arrow, red!70!black] (l3) -- (deny3);
  \end{tikzpicture}
  \caption{The \armname{} layered policy pipeline.  Each layer can independently
  deny a tool call.}
  \label{fig:pipeline}
\end{figure}

\subsubsection{Layer 1: Hard Boundaries}
\label{sec:l1}

Layer~1 captures unconditional boundary checks that do not depend on session
provenance or operator policy.  It enforces five invariants:

\begin{description}
  \item[HB-1: Payload size limit.] Any string argument exceeding 10{,}000
    characters is denied (bulk exfiltration defense).
  \item[HB-2: Sensitive path blocking.] File-path arguments matching protected
    patterns (\texttt{\~{}/.ssh/*}, \texttt{/etc/shadow}, \texttt{\~{}/.aws/*})
    are denied.
  \item[HB-3: Call-rate limit.] Any tool called more than 100 times per session
    is denied (runaway agent defense).
  \item[HB-4: Credential detection.] Arguments containing patterns matching
    known credential formats (AWS access keys, GitHub tokens, JWTs, RSA private
    keys, etc.) are denied.
  \item[HB-5: Schema pinning.] The SHA-256 hash of each tool's input schema is
    recorded on first invocation.  Subsequent calls with a changed schema are
    denied (rug-pull attack defense~\citep{invariantlabs2024mcp}).
\end{description}

\subsubsection{Layer 2: Graph-Aware Provenance}
\label{sec:l2}

Layer~2 queries the provenance graph to determine whether the current tool
call's inputs can be traced back to trusted sources.  It performs two checks:
trust propagation, which computes the minimum trust level across all ancestor
nodes of the current call, and counterfactual chain detection, which checks
whether any path from a \textsc{DeniedAction} node reaches the current call via
a \textsc{Counterfactual} edge.  If the minimum reachable trust falls below the
configured threshold, or if a denial-induced counterfactual path reaches the
current call, Layer~2 returns \textsc{Deny}.  The formal semantics of the
provenance graph, including node types, edge types, and the trust propagation
algorithm, are defined in \Cref{sec:provenance-graph}.

\subsubsection{Layer 3: Schema-Derived Constraints}
\label{sec:l3}

Layer~3 derives lightweight defensive constraints automatically from tool input
schemas, reducing reliance on manual policy specification.  It applies four rules:

\begin{description}
  \item[SD-1: Suspicious parameter names.]  Parameters whose names suggest
    free-text fields (e.g., \texttt{note}, \texttt{comment},
    \texttt{metadata}, \texttt{description}) receive a 100-character length
    limit and credential-pattern checks, as these are common exfiltration
    side channels.
  \item[SD-2: Filesystem parameters.]  Parameters named \texttt{path},
    \texttt{file}, \texttt{directory} receive sensitive-path blocking.
  \item[SD-3: Optional parameters.]  Optional string parameters not listed in
    the schema's \texttt{required} array receive a 50-character length limit
    (optional parameters are the most common exfiltration vector).
  \item[SD-4: Network parameters.]  Parameters named \texttt{url},
    \texttt{endpoint}, \texttt{host} are flagged for logging (enforcement
    deferred to future work).
\end{description}

\subsubsection{Layer 4: Manual Policy}
\label{sec:l4}

Layer~4 enforces operator-defined rules specified in an immutable capability
token.  Rules include per-tool allow/deny, call budgets, argument blocklists, and
value constraints.  This layer provides the traditional capability-based access
control that operators expect.  Layer~4 evaluates these rules against the
capability token bound to the current session.

\subsection{Capability Tokens}
\label{sec:capability-tokens}

Each agent session is bound to an immutable \emph{capability token} that specifies
permitted tools and their constraints.  Tokens are frozen dataclasses: the agent
cannot modify them at runtime.  This design draws on capability-based
security~\citep{dennis1966programming} and
Macaroons~\citep{birgisson2014macaroons}, where capabilities can be attenuated
(restricted) but never amplified.  Layer~4 evaluates operator-defined rules
against the capability token bound to the current session.

\subsection{Audit Log}
\label{sec:audit-log}

Every enforcement decision (\textsc{Allow} or \textsc{Deny}) is recorded in a
hash-chained audit log.  Each entry contains the tool name, arguments, decision,
reason, and per-layer results.  The hash chain provides tamper evidence: modifying
any entry invalidates all subsequent hashes.  The audit log therefore captures
both successful and denied tool invocations at the same enforcement boundary.
This supports forensic analysis and compliance requirements.

%% file: sections/05-provenance-graph.tex
The provenance graph is the core data structure behind \armname{}'s graph-aware
enforcement.  Its purpose is not merely to log events, but to represent the
causal structure that flat provenance models miss: successful data flow,
field-level provenance, and denial-induced downstream influence.  A
single label attached to a tool result cannot distinguish which field carried
taint, and a trace of successful calls alone cannot represent the effect of a
denied action.  The provenance graph provides a single representation over
successful outputs, structured fields, and denied actions, allowing Layer~2 to
reason over both explicit data dependencies and counterfactual links introduced
by denials.

\subsection{Graph Structure}

\begin{definition}[Provenance Graph]
\label{def:provenance-graph}
A provenance graph $G = (V, E, \tau, \ell)$ consists of:
\begin{itemize}
  \item A set of nodes $V$ partitioned into four disjoint subsets:
    \[
      V = V_{\textsc{Call}} \;\cup\; V_{\textsc{Data}} \;\cup\;
      V_{\textsc{DataField}} \;\cup\; V_{\textsc{DeniedAction}},
    \]
    representing tool calls, returned data items, sub-items of structured
    results, and denied tool calls, respectively.
  \item A set of directed edges $E \subseteq V \times V$ with labels from
    $\{$\textsc{DirectOutput}, \textsc{InputTo}, \textsc{Counterfactual},
    \textsc{FieldOf}$\}$.
  \item A trust assignment $\tau: V_{\textsc{Data}} \cup V_{\textsc{DataField}}
    \to \mathcal{T}$ mapping data nodes to trust levels from the integrity
    lattice (\Cref{def:integrity-lattice}).
  \item A label function $\ell$ assigning metadata such as tool names,
    arguments, timestamps, and denial reasons to each node.
\end{itemize}
\end{definition}

\subsection{Edge Semantics}

Each edge type captures a specific provenance or causal relationship.  All edges are directed;
the notation $\textsc{Label}(u, v)$ denotes an edge from $u$ to $v$.

\begin{description}
  \item[\textsc{DirectOutput}$(c, d)$:]  Tool call $c \in V_{\textsc{Call}}$
    produced data item $d \in V_{\textsc{Data}}$.  This is the primary
    data-flow edge: it records which call generated which output.

  \item[\textsc{InputTo}$(d, c)$:]  Data item $d \in V_{\textsc{Data}}$ was
    used as an argument to tool call $c \in V_{\textsc{Call}}$.  Together,
    \textsc{DirectOutput} and \textsc{InputTo} define transitive provenance
    chains.  For example, if
    \[
      c_1 \xrightarrow{\textsc{DirectOutput}} d_1
      \xrightarrow{\textsc{InputTo}} c_2
      \xrightarrow{\textsc{DirectOutput}} d_2
      \xrightarrow{\textsc{InputTo}} c_3\,,
    \]
    then $c_3$ is transitively influenced by the outputs of $c_1$.

  \item[\textsc{FieldOf}$(f, d)$:]  Field node $f \in V_{\textsc{DataField}}$
    is a component of structured data item $d \in V_{\textsc{Data}}$.  This
    enables field-level trust assignment: given a tool result with fields
    \texttt{name} and \texttt{email}, the integrator can assign
    \textsc{ToolTrusted} to \texttt{name} and \textsc{ToolUntrusted} to
    \texttt{email} independently.

  \item[\textsc{Counterfactual}$(a_d, c)$:]  Denied action
    $a_d \in V_{\textsc{DeniedAction}}$ is linked to subsequent tool call
    $c \in V_{\textsc{Call}}$ by a counterfactual-dependence edge.  This edge records the
    potential causal influence of a denial on later behavior and is the
    mechanism for detecting causality laundering
    (\Cref{sec:causality-laundering}).
\end{description}

\subsection{Trust Propagation}

The provenance graph enables a conservative trust computation: a tool call's
effective trust level is the minimum trust of any data ancestor in its causal
history.  One untrusted ancestor anywhere in the chain is sufficient to lower the
trust of every downstream node.  This is why the graph matters operationally:
flat labels attached only to immediate tool results cannot capture transitive
taint that propagates across multiple tool calls.

\begin{definition}[Minimum Reachable Trust]
\label{def:min-trust}
For a node $v \in V$, define:
\[
  \textsc{MinTrust}(v) =
  \min_{u \in \textsc{Ancestors}(v) \cap
  (V_{\textsc{Data}} \cup V_{\textsc{DataField}})} \tau(u)
\]
where $\textsc{Ancestors}(v)$ is the set of all nodes reachable from $v$ by
following edges in reverse.  If $\textsc{Ancestors}(v)$ contains no data nodes,
then $\textsc{MinTrust}(v) = \textsc{SysInstr}$.
\end{definition}

This is a worst-case lower bound: if \emph{any} data ancestor is untrusted, the
effective trust of the current node is bounded by that ancestor's level.  The
computation is analogous to taint propagation in dynamic taint
analysis~\citep{schwartz2010all}, but it operates over a causal provenance graph
rather than instruction-level data flow.

\begin{property}[Monotonic Taint]
For any edge $(u, v) \in E$,
\[
  \textsc{MinTrust}(v) \leq \textsc{MinTrust}(u).
\]
\end{property}

\subsection{Enforcement Queries}

The graph semantics above support two enforcement queries, corresponding to the
two checks performed by Layer~2 (\Cref{sec:l2}) when deciding whether to block a
tool call $c$:

\begin{enumerate}
  \item $\textsc{MinTrust}(c) < \theta$, where $\theta$ is the minimum required
    trust level (default: \textsc{ToolTrusted}).  This catches transitive taint
    chains.

  \item $\textsc{CounterfactualChains}(c) \neq \emptyset$, where
    $\textsc{CounterfactualChains}(c)$ denotes the set of paths to $c$
    containing at least one \textsc{Counterfactual} edge.  This catches
    causality laundering.
\end{enumerate}

Both queries are implemented via graph traversal on a \texttt{rustworkx}
backend~\citep{rustworkx2024}, achieving sub-millisecond latency for typical
session graphs (tens to hundreds of nodes).

\subsection{Graph Construction}

The graph is constructed incrementally as the agent session progresses.
\armname{}'s \texttt{GraphAwareEngine} orchestrates construction at two points:
before tool execution, when the call node and its input dependencies are
materialized, and after successful execution, when returned outputs are added to
the graph.  Algorithm~\ref{alg:graph-eval} shows the pre-execution evaluation
path.

\begin{algorithm}[t]
\caption{Pre-Execution Graph-Aware Tool Call Evaluation}
\label{alg:graph-eval}
\begin{algorithmic}[1]
\Require Tool name $t$, arguments $\mathbf{a}$, input data IDs $D_{\text{in}}$
\Ensure Policy decision $\delta \in \{\textsc{Allow}, \textsc{Deny}\}$
\State $c \gets \textsc{NewNode}(V_{\textsc{Call}}, t, \mathbf{a})$
  \Comment{Create call node}
\For{$d \in D_{\text{in}}$}
  \State $\textsc{AddEdge}(d, c, \textsc{InputTo})$
    \Comment{Link input data}
\EndFor
\If{$\exists\, a_d \in V_{\textsc{DeniedAction}}$ from preceding denial}
  \State $\textsc{AddEdge}(a_d, c, \textsc{Counterfactual})$
    \Comment{Auto-link}
\EndIf
\State $\delta \gets \textsc{EvaluatePipeline}(t, \mathbf{a}, c)$
  \Comment{L1 + L2G + L3 + L4}
\If{$\delta = \textsc{Deny}$}
  \State $a_d \gets \textsc{NewNode}(V_{\textsc{DeniedAction}}, t, \mathbf{a},
    \delta.\text{reason})$
\EndIf
\State \Return $\delta$
\end{algorithmic}
\end{algorithm}

If the call is allowed and the upstream tool returns successfully, \armname{}
creates one or more $V_{\textsc{Data}}$ nodes for the returned output, adds
\textsc{DirectOutput} edges from the call node to those data nodes, and, if the
output is structured, optionally materializes $V_{\textsc{DataField}}$ nodes
linked by \textsc{FieldOf} edges.  Trust labels are assigned at this point
according to the source trust level and any field-level overrides.

\subsection{Field-Level Provenance}

Structured tool outputs (JSON objects, database rows) often contain fields with
heterogeneous trust levels.  \armname{} decomposes such outputs into
\textsc{DataField} nodes with independent trust assignments:

\begin{definition}[Field-Level Trust Override]
\label{def:field-trust}
Given a data item $d$ with structured value $\{k_1: v_1, \ldots, k_n: v_n\}$
and a trust override map $\omega: K \rightharpoonup \mathcal{T}$, the trust of
field $f_i$ is:
\[
  \tau(f_i) = \begin{cases}
    \omega(k_i) & \text{if } k_i \in \text{dom}(\omega) \\
    \tau(d)     & \text{otherwise}
  \end{cases}
\]
\end{definition}

This enables fine-grained enforcement: a contact record's \texttt{name} field
may be trusted while its \texttt{email} field is untrusted (e.g., externally
provided).  A subsequent \texttt{send\_email(to=email)} call would be denied
because the \texttt{email} field's taint propagates through the graph.

%% file: sections/06-implementation.tex
We implemented \armname{} as a Python prototype to validate that the
architecture described in \Cref{sec:architecture,sec:provenance-graph} can be
realized with low overhead and integrated into existing agent frameworks without
modifying the LLM, the tool servers, or the prompt templates.  The implementation
is split into two modules: \texttt{arm\_core}, which contains the control-plane
enforcement pipeline (policy layers, capability tokens, audit log, MCP proxy),
and \texttt{arm\_provenance}, which contains the provenance graph engine and the
graph-aware enforcement layer.  The total enforcement codebase is approximately
910~lines of Python, which keeps the trusted enforcement path relatively
compact.

\subsection{Module Architecture}

\begin{table}[t]
\centering
\caption{Implementation modules and sizes.}
\label{tab:modules}
\begin{tabular}{@{}lll@{}}
\toprule
\textbf{Module} & \textbf{Component} & \textbf{Lines} \\
\midrule
\texttt{arm\_core}
  & PolicyEngine, LayeredPolicyEngine  & $\sim$130 \\
  & CapabilityToken, ToolPermission    & $\sim$50 \\
  & L1: HardBoundariesLayer            & $\sim$120 \\
  & L2: ProvenanceLayer (flat baseline) & $\sim$110 \\
  & L3: SchemaDerivedLayer             & $\sim$80 \\
  & L4: ManualPolicyLayer              & $\sim$60 \\
  & AuditLog                           & $\sim$60 \\
  & MCP Proxy                          & $\sim$50 \\
\midrule
\texttt{arm\_provenance}
  & ProvenanceGraph                    & $\sim$430 \\
  & GraphProvenanceLayer (L2G)         & $\sim$60 \\
  & GraphAwareEngine                   & $\sim$110 \\
\midrule
\multicolumn{2}{@{}l}{\textbf{Total enforcement code}} & $\sim$910 \\
\multicolumn{2}{@{}l}{\textbf{Test code}} & $\sim$600 \\
\bottomrule
\end{tabular}
\end{table}

Most of the implementation complexity lies in the provenance subsystem
(\texttt{arm\_provenance}), reflecting that the implementation complexity is
concentrated in graph construction and graph-aware enforcement rather than in
the policy pipeline itself.  The flat baseline provenance layer is retained only
to support the comparative evaluation in \Cref{sec:evaluation}.

\subsection{Key Design Decisions}

\paragraph{Deterministic over probabilistic.}  Every enforcement decision in
\armname{} is deterministic.  The graph-aware provenance layer (L2G) queries the
provenance graph directly, with no LLM calls in the enforcement loop.  This is a
deliberate design choice: a defense that relies on the same LLM under attack to
validate provenance or policy compliance inherits that model's compromise
surface.

\paragraph{rustworkx for graph operations.}  The provenance graph is backed by
\texttt{rustworkx}~\citep{rustworkx2024}, a Rust-implemented graph library with
Python bindings.  Reachability queries are implemented on this backend, achieving
sub-millisecond latency for session-sized graphs.  This ensures \armname{}
introduces negligible overhead compared to LLM inference time (typically
100\,ms to 10\,s).

\paragraph{MCP protocol-level interposition.}  \armname{} operates as an MCP
proxy server: it presents the same tool schemas to the LLM client while
intercepting every \texttt{tool\_use} request.  This achieves complete mediation
without modifying the LLM, the MCP client library, or the upstream tool
servers.  While our prototype uses MCP for concreteness, the same interposition
pattern applies to equivalent tool-calling runtimes with a centralized
execution boundary.

\paragraph{Immutable capability tokens.}  Capability tokens are implemented as
Python frozen dataclasses.  Once created, neither the agent nor any tool can
modify a token's permissions.  This provides a strong static guarantee: the
agent's capability surface is fixed at session initialization.

\paragraph{Hash-chained audit.}  Each audit entry includes the SHA-256 hash of
the previous entry.  Chain verification detects any post-hoc modification,
supporting forensic analysis and regulatory compliance.

\subsection{Integration Points}

\armname{} is designed as a drop-in enforcement layer for existing agent
frameworks.  In practice, it assumes only that tool invocations pass through a
centralized execution boundary and that successful tool results can be recorded
after execution.  Integration requires:

\begin{enumerate}
  \item For MCP-based agents, routing tool calls through \armname{}'s MCP
    proxy.  For programmatic agents, calling
    \texttt{engine.evaluate()} before each tool execution.
  \item Providing an \texttt{EvaluationContext} with user messages, tool-call
    history, tool schemas, and capability token.
  \item For graph-aware enforcement: calling
    \texttt{engine.record\_tool\_result()} after each allowed tool call to
    update the provenance graph.
\end{enumerate}

This integration does not require changes to the LLM, the prompt templates, or
the upstream tool implementations.

%% file: sections/07-evaluation.tex
The central claim of this paper is that graph-aware provenance tracking catches
security-relevant attacks that a flat provenance baseline misses, and that the
added expressiveness does not impose meaningful runtime cost.  We evaluate this
claim through a controlled differential experiment: three attack scenarios are
executed against two engine configurations that differ \emph{only} in the
provenance layer (L2 vs.\ L2G), with all other layers held constant.  The
evaluation is scenario-based and targets the prototype implementation; we
explicitly scope what it does and does not establish in
\Cref{subsec:limitations}.

\subsection{Experimental Setup}

Each scenario is executed against two engine configurations:

\begin{itemize}
  \item \textbf{Flat baseline:} \texttt{LayeredPolicyEngine} with the
    citation-based provenance layer~(L2).  L2 checks whether each argument
    value appears somewhere in the agent's tool-call history; if the argument
    can be grounded in prior tool-call history, the provenance check passes.
  \item \textbf{Graph-aware:} \texttt{GraphAwareEngine} with the graph-aware
    provenance layer~(L2G).  L2G queries the provenance graph for trust
    propagation, counterfactual edges, and field-level provenance.
\end{itemize}

Both configurations share identical L1 (hard boundary), L3 (schema-derived),
and L4 (manual policy) layers.  Because only the provenance mechanism differs,
any divergence in outcome is attributable to the graph-aware representation.

A scenario is \emph{successfully defended} if the final sink-bound action is
denied before execution.  The three scenarios were selected to isolate the three
structural capabilities that the graph-aware model adds beyond the flat
baseline: counterfactual-edge detection, transitive trust propagation, and
field-level provenance.  Each scenario is constructed so that one of these
capabilities is the decisive difference between the two engines.

\subsection{Attack Scenarios}

Each of the following scenarios isolates one structural advantage of the
graph-aware model over the flat baseline.

\subsubsection{Attack 1: Causality Laundering (Denial Inference Exfiltration)}

The attacker's payload instructs the agent to:
\begin{enumerate}
  \item Call \texttt{read\_file("/etc/shadow")} (denied by L1, HB-2).
  \item Call \texttt{send\_email(body="shadow file exists")}.
\end{enumerate}

The second call contains no credential patterns, no sensitive paths, and its
14-character body passes all length checks.  The string does not appear in any
tool output; it is composed by the LLM from the denial signal.

\begin{itemize}
  \item \textbf{Flat result: \textsc{Allow}.}  L2 finds no provenance
    violation, because the sink-bound string has no explicit returned-data
    lineage and does not violate any non-provenance boundary checks.
  \item \textbf{Graph result: \textsc{Deny}.}  L2G detects a
    \textsc{Counterfactual} edge from the denied \texttt{read\_file} to the
    \texttt{send\_email} call.  Reason: ``causality laundering detected.''
\end{itemize}

This scenario isolates the value of denial-aware counterfactual tracking.

\subsubsection{Attack 2: Transitive Taint Chain}

The attacker's payload instructs the agent to:
\begin{enumerate}
  \item Call \texttt{read\_emails()} $\rightarrow$ returns
    \texttt{\{body: "Wire \$50k to account XXX"\}} (untrusted).
  \item Call \texttt{format\_response(text="Summary: Wire transfer...")} to
    transform the content.
  \item Call \texttt{send\_email(body=...)} to \texttt{finance@company.com}
    with the formatted text.
\end{enumerate}

The intermediate \texttt{format\_response} call launders the data's provenance:
its output appears in the agent's tool-call history as an allowed tool result,
obscuring the earlier untrusted origin.

\begin{itemize}
  \item \textbf{Flat result: \textsc{Allow}.}  L2 finds
    \texttt{format\_response} output in the tool-call history and the citation
    check passes.  The flat layer has no mechanism to look beyond the immediate
    source; the transitive origin in \texttt{read\_emails} is invisible.
  \item \textbf{Graph result: \textsc{Deny}.}  L2G traces the provenance graph:
    \texttt{send\_email} depends on \texttt{format\_response}, which depends on
    \texttt{read\_emails} (untrusted).  The resulting \textsc{MinTrust} is
    \textsc{ToolUntrusted}, below the threshold.
\end{itemize}

This scenario isolates the value of transitive trust propagation through the
provenance graph.

\subsubsection{Attack 3: Mixed-Provenance Field Exploit}

A tool returns a structured result whose fields do not share the same trust
level:
\begin{enumerate}
  \item Call \texttt{get\_contact()} $\rightarrow$ returns
    \texttt{\{name: "Alice", email: "attacker@evil.com"\}}.  The \texttt{name}
    field is trusted; the \texttt{email} field is marked untrusted (externally
    provided).
  \item Call \texttt{send\_email(to="attacker@evil.com")}.
\end{enumerate}

\begin{itemize}
  \item \textbf{Flat result: \textsc{Allow}.}  L2 finds
    \texttt{"attacker@evil.com"} in the tool-call history (from
    \texttt{get\_contact}), so the citation check passes.  The flat layer
    treats the returned object as a single unit and cannot distinguish trusted
    fields from untrusted ones within the same value.
  \item \textbf{Graph result: \textsc{Deny}.}  L2G resolves the \texttt{email}
    field node individually, finding trust \textsc{ToolUntrusted}.  The
    \texttt{name} field, by contrast, would be allowed independently.
\end{itemize}

This scenario isolates the value of field-level provenance within
structured tool outputs.

\subsection{Results}

\begin{table}[t]
\centering
\caption{Security evaluation: flat (L2 citation) vs.\ graph-aware (L2G)
  enforcement.  \cmark{} = attack blocked, \xmark{} = attack missed.}
\label{tab:eval-results}
\begin{tabular}{@{}lccc@{}}
\toprule
\textbf{Attack} & \textbf{Flat (L2)} & \textbf{Graph (L2G)} & \textbf{Root Cause} \\
\midrule
Causality laundering    & \xmark & \cmark & Counterfactual edge \\
Transitive taint chain  & \xmark & \cmark & Graph reachability \\
Mixed-provenance fields & \xmark & \cmark & Field-level provenance \\
\bottomrule
\end{tabular}
\end{table}

\Cref{tab:eval-results} summarizes the results.  Because all non-provenance
layers (L1, L3, L4) are held constant across configurations, the difference in
outcomes is attributable solely to the provenance mechanism.

The flat baseline misses all three attacks because it verifies only that a value
\emph{appeared} in the agent's history, not \emph{how it arrived there} or
\emph{which specific field} contributed it.  The graph-aware engine blocks all
three by exploiting structural properties of the provenance graph:
counterfactual edges (Attack~1), transitive reachability (Attack~2), and
field-level provenance (Attack~3).  Each blocked attack corresponds to
a capability that the flat representation does not capture.

\subsection{Performance}

On the prototype implementation, the complete evaluation pipeline (L1 + L2G +
L3 + L4) executes in under 1\,ms for all three scenarios (Apple M4 Pro,
48\,GB RAM; median of 100 runs per scenario after warm-up).
For context, a single LLM inference call typically takes 100\,ms to several
seconds; the enforcement overhead is therefore negligible relative to the
latency already present in any tool-calling agent loop.  These measurements
reflect session-sized graphs (tens of tool calls, hundreds of data nodes) and
should not be extrapolated to arbitrarily large provenance histories without
further benchmarking.

Memory scales linearly with session length: a typical session produces a graph
with $O(100)$ nodes, well within practical limits.  Reachability queries run in
$O(V + E)$ time via the \texttt{rustworkx} backend.

\subsection{Test Suite}

In addition to the scenario-based evaluation above, \armname{} includes 45 unit
tests covering:
\begin{itemize}
  \item Provenance graph construction, trust lattice, reachability, and
    counterfactual edges (34 tests).
  \item Graph-aware engine end-to-end scenarios: legitimate calls, denial
    inference, field-level taint, and composition with L1 (11 tests).
\end{itemize}

All tests are deterministic and execute without network access or LLM
invocation, enabling CI integration.

\subsection{Limitations}
\label{subsec:limitations}

The flat baseline is intentionally a minimal flat-history provenance mechanism
designed to isolate the representational contribution of the graph-aware model;
it is not a reimplementation of FIDES, PCAS, or other full systems.  The
current evaluation uses manually constructed attack scenarios rather than an
automated benchmark suite (e.g., AgentDojo~\citep{debenedetti2024agentdojo}).
While the three scenarios are representative of documented attack patterns, a
larger-scale evaluation is needed to assess false-positive rates, coverage
across diverse agentic workflows, and performance under longer provenance
histories.  We also do not yet report benchmark-scale end-to-end experiments
with frontier LLMs in the loop.  We discuss these limitations and corresponding
future work in \Cref{sec:discussion}.

%% file: sections/08-related-work.tex
Closest prior work on runtime security for tool-calling agents spans
information-flow control, graph-based dependency enforcement, secure-by-design
agent architectures, and causal defenses against indirect prompt injection.
FIDES~\citep{costa2025fides}, PCAS~\citep{palumbo2026pcas}, and
CaMeL~\citep{debenedetti2025defeating} are the closest prior systems to \armname{},
each on a different axis.  \armname{} does not claim to be the first runtime
security mechanism for LLM agents, nor the first graph-based or causal defense.
Its contribution is narrower: it makes denied tool invocations first-class
provenance events and enforces over their downstream influence through
counterfactual edges, enabling detection of denial-feedback leakage patterns
that are not naturally captured by successful-flow provenance, extracted
control/data flow, replay-based causal influence, or structured runtime traces.

\paragraph{Deterministic runtime enforcement.}
FIDES and PCAS are the two closest prior systems to \armname{}, but on different
axes.  FIDES~\citep{costa2025fides} applies information-flow control to AI
agents, propagating confidentiality and integrity labels over observed flows and
enforcing policy deterministically.  Both FIDES and \armname{} aim for
deterministic runtime enforcement over agent behavior, but FIDES tracks labels
over successful flows, whereas \armname{} additionally represents denied actions
and enforces over their downstream influence.
PCAS~\citep{palumbo2026pcas} compiles declarative authorization policies into an
instrumented agent runtime that enforces them over a dependency graph via
reference-monitor interposition.  PCAS is the closest prior work to \armname{}
on the graph-based enforcement substrate: it provides deterministic, graph-based policy
enforcement without requiring security-specific architectural restructuring.
The key distinction is that neither system explicitly models denied actions as
provenance-bearing events or enforces over their downstream counterfactual
influence.  FIDES propagates labels over data that successfully
flowed through the agent; PCAS enforces policies over dependency relationships
among completed actions.  \armname{} extends provenance with denied-action nodes
and counterfactual edges so that denial-feedback leakage (where the LLM infers
sensitive information from the \emph{absence} of a successful result) becomes
enforceable.  That distinction is central to \armname{}'s treatment of causality
laundering (\Cref{sec:causality-laundering}).

\paragraph{Secure-by-design architectures.}
CaMeL~\citep{debenedetti2025defeating} takes a complementary approach: it
extracts trusted control and data flow from the user query, prevents untrusted
observations from steering execution, and uses capabilities to restrict
unauthorized data exfiltration.  CaMeL's strength is architectural prevention;
it changes how the agent is built.  \armname{} occupies a different design
point: it assumes a conventional tool-calling agent and hardens the execution
boundary through provenance-aware runtime mediation.  CaMeL does not center
denied-action provenance or field-level provenance within structured
tool outputs, while \armname{} does not require a planner/executor split or a
new programming model.  The two approaches are complementary; an agent built
with CaMeL's architecture could additionally deploy \armname{} at the tool
boundary for defense in depth.

\paragraph{Causal defenses against indirect prompt injection.}
A growing line of work applies causal reasoning to detect and mitigate indirect
prompt injection.  AgentSentry~\citep{zhang2026agentsentry} models multi-turn
indirect prompt injection as temporal causal takeover, localizes takeover points
through controlled counterfactual re-execution at tool-return boundaries, and
mitigates attacks via context purification.
CausalArmor~\citep{kim2026causalarmor} uses leave-one-out ablation to attribute
privileged actions to user intent versus untrusted observations and triggers
sanitization only when an untrusted segment dominates.  Both systems reason
about \emph{why} an action was taken, rather than only whether an input appears
suspicious; however, they do so through replay or attribution at inference time,
requiring additional LLM calls to diagnose causal influence.
\armname{} is complementary to these approaches.  Rather than replaying the
agent or estimating causal influence online, \armname{} records a conservative
approximation of denial-induced influence directly inside the provenance graph
through denied-action nodes and counterfactual edges.  This keeps enforcement
deterministic and lightweight, with no additional LLM calls, but also makes it
more conservative: \armname{} may flag actions that a replay-based system would clear
after re-execution.  The trade-off is between precision (replay) and guaranteed
low-latency deterministic enforcement (graph-native provenance).  Recent
evidence that frontier agents can infer hidden monitoring from blocking feedback
alone~\citep{jiralerspong2026noticing} further supports the practical relevance
of denial-feedback channels as an attack surface.

\paragraph{Structured runtime analysis.}
AgentArmor~\citep{wang2025agentarmor} also rejects flat history-based analysis
and instead converts runtime traces into structured program-dependence
representations (control-flow graphs, data-flow graphs, and program-dependence
graphs), then applies type-system-style analysis to detect unsafe behavior.
\armname{} shares the motivation of moving beyond flat message histories into
graph-structured reasoning, but is narrower and more enforcement-oriented:
\armname{}'s provenance graph is defined over tool calls, tool outputs,
structured data fields, and denied invocations, and its queries are designed for
real-time enforcement decisions rather than post-hoc program analysis.

\paragraph{Adjacent runtime governance and complementary defenses.}
Several recent systems address adjacent aspects of agent runtime security.
RTBAS~\citep{zhong2025rtbas} adapts information-flow control to tool-based
agents with a mix of automatic safe execution and human confirmation.
ToolSafe~\citep{mou2026toolsafe} provides step-level pre-execution safety
detection.  AgentGuardian~\citep{abaev2026agentguardian} learns context-aware
access-control policies from execution traces.
PRISM~\citep{li2026prism} provides a defense-in-depth runtime layer for
deployable agent gateways, and OPP~\citep{zhu2026opp} proposes protocol-level
governance for agent-to-tool communication.  At the prompt level,
Spotlighting~\citep{hines2024defending} and instruction
hierarchy~\citep{wallace2024instruction} harden the LLM against prompt
injection, while NeMo Guardrails~\citep{nvidia2024nemoguardrails} provides
configurable input/output filtering.  These are complementary to \armname{}:
they reduce injection probability or restrict tool access, while \armname{}
limits damage when injection succeeds by enforcing provenance invariants at the
execution boundary.

\paragraph{Foundations.}
\armname{} draws on classical work in reference monitors, information-flow
control, and causality.  Anderson's reference monitor
concept~\citep{anderson1972planning} and Saltzer and Schroeder's design
principles~\citep{saltzer1975protection} provide the model for complete
mediation and tamper-resistant enforcement.  Denning's lattice
model~\citep{denning1976lattice} provides the theoretical basis for
\armname{}'s integrity lattice and trust ordering.  Pearl's counterfactual
causality framework~\citep{pearl2009causality} clarifies why denial feedback can
act as an implicit information channel: an agent that observes a denied action
gains information about the protected resource, and that information can
causally influence subsequent actions.  Database
provenance~\citep{green2007provenance,cheney2009provenance} motivates
graph-based dependency tracking at a finer granularity than message histories.
\armname{}'s main conceptual extension over these foundations is to treat denied
tool invocations as provenance-bearing events whose downstream consequences are
themselves policy-relevant.

%% file: sections/09-discussion.tex
\subsection{What \armname{} Does and Does Not Claim}

\armname{} is not a complete solution to agent security, nor a general causal
attribution engine.  Its contribution is narrower: it is a deterministic
execution-boundary monitor that enriches runtime provenance with denied actions
and counterfactual edges, enabling enforcement over a denial-feedback leakage
channel that successful-flow provenance alone does not naturally capture.

Accordingly, our formal and empirical claims should be read at that level of
specificity.  Theorem~\ref{thm:mediated-integrity} is supported by the
prototype implementation in the following sense: under the deployment model,
every tool call traverses the layered enforcement pipeline, and graph-aware
provenance blocks the evaluated attack patterns before execution.  However, the
complete-mediation property depends on correct deployment, namely routing all
tool calls through the proxy.  This is an operational assumption, not a property
that code alone can guarantee.  More precisely,
Theorem~\ref{thm:mediated-integrity} formalizes mediated enforcement under the
deployment model; it does not establish completeness for all implicit-flow or
denial-feedback attacks.

Two additional limitations follow directly from the current mechanism.  First,
the counterfactual heuristic is a sufficient condition for detecting the
dominant probe-then-exfiltrate pattern, not a necessary condition for detecting
all denial-feedback leakage.  Second, field-level protection depends on
integrator-provided trust overrides for structured tool outputs.  If all fields
inherit the parent object's trust level, field-level provenance adds no
security beyond object-level labeling.

\subsection{Counterfactual Heuristic: False Positives and Missed Chains}

\armname{} links a denied action to the immediately following tool call via a
counterfactual edge.  This heuristic is intentionally conservative because the
LLM's internal reasoning is opaque: the enforcement layer cannot directly
observe whether a later action was mentally conditioned on the denial, so it
approximates causal influence using temporal adjacency.

This approximation has two consequences.  It can over-approximate by flagging a
benign tool call that happens to follow a denial, and it can under-approximate
by missing delayed or multi-step laundering chains in which denial-derived
information is transformed across several subsequent actions before reaching a
sink.  We therefore do not claim completeness for denial-feedback leakage
detection.

The case for the heuristic is pragmatic rather than semantic.  First, it makes a
previously unrepresented attack surface, denial-induced influence, explicit in
the provenance model.  Second, prompt-injection exfiltration often exhibits
temporal locality: probe and exfiltration tend to occur close together in the
same tool-use trajectory.  A more precise alternative would require structural
causal analysis~\citep{pearl2009causality} or selective replay, but that would
sacrifice the current design's deterministic, low-latency enforcement loop.

\subsection{Trust Assignment and Declassification}

\armname{}'s trust propagation is monotone by design.  A tool call inherits the
minimum trust level of any reachable data ancestor, which makes the mechanism a
worst-case lower bound over provenance rather than a semantic judgment about
content quality.  This conservative design is attractive because it is simple,
deterministic, and easy to audit, but it also shifts burden onto trust
assignment at ingestion time.

This matters most for structured outputs.  Field-level provenance is
useful only when the integrator provides meaningful per-field overrides.  If all
fields inherit the parent trust level, then field-level provenance collapses
back to ordinary object-level tainting.  In that sense, the mechanism is only as
good as the trust labels attached to tool outputs.

A related open problem is declassification.  Real workflows may require a
trusted validator or sanitization stage that upgrades a value's trust after
explicit checking.  Supporting this safely is difficult because the declassifier
must itself be outside the LLM trust boundary and resistant to manipulation.
One promising direction is to treat declassification as an auditable capability
(inspired by Macaroons~\citep{birgisson2014macaroons}), bound to specific
values or transformations rather than granted globally.

\subsection{Evaluation Scope}

Our evaluation is a controlled differential experiment, not a benchmark-scale
security evaluation.  The three scenarios were chosen to isolate the exact
representational advantages of the graph-aware model over the flat baseline:
denied-action provenance, transitive taint propagation, and field-level
provenance.  This is sufficient to validate the paper's central claim that these
structural features matter, but it does not establish false-positive rates,
coverage across diverse agent workflows, or behavior under long provenance
histories.

A stronger evaluation should include benchmark suites such as
AgentDojo~\citep{debenedetti2024agentdojo} and
InjectAgent~\citep{zhan2024injectagent}, longer-running trajectories, and
experiments with frontier tool-calling agents in the loop.  Those studies are
important future work, but they would extend rather than replace the present
controlled comparison.

\subsection{Multi-Agent Extension}

The current paper studies a single-agent setting.  Extending the provenance
graph to multi-agent workflows is plausible but non-trivial because delegation,
concurrency, and composition introduce new causal relationships that are absent
in a single-agent session.  Useful building blocks include delegation tokens
with monotonic attenuation, causal ordering mechanisms such as vector clocks,
and composition checks that verify whether agent-local guarantees survive across
a workflow of interacting agents.

We view that extension as a follow-up problem rather than part of the present
claim.  The contribution of this paper is the denied-action provenance mechanism
in the single-agent case; multi-agent composition should be evaluated on its own
terms.

%% file: sections/10-conclusion.tex
Tool-calling LLM agents can read private data, invoke external services, and
trigger real-world side effects.  In such systems, the critical security
boundary is the tool-execution interface: every tool call is an access decision,
and every tool result can introduce untrusted influence into later actions.
Securing these systems therefore requires more than prompt hardening or post-hoc
detection; it requires principled runtime enforcement at the execution boundary.

This paper presented \armname{}, a deterministic, provenance-aware enforcement
layer for tool-calling agents.  \armname{}'s central contribution is to treat
denied tool invocations as first-class provenance events and to represent their
possible downstream influence through counterfactual edges.  That extension
makes a denial-feedback leakage channel enforceable: downstream actions can be
blocked not only when they depend on untrusted successful flows, but also when
they are causally shaped by a prior denial.  The same graph-aware mechanism also
supports transitive taint propagation and field-level provenance within
structured tool outputs.

In a controlled differential evaluation, the graph-aware engine blocked
causality laundering, transitive taint propagation, and mixed-provenance field
misuse that the flat provenance baseline missed, while adding negligible runtime
overhead relative to ordinary agent execution.  \armname{} is not a complete
solution to agent security, but these results suggest that denial-aware causal
provenance is a useful abstraction for runtime enforcement in tool-calling agent
systems.

%% file: sections/appendix-proofs.tex
\begin{theorem}[Mediated Integrity]
\label{thm:mediated-integrity}
In an \armname{}-protected system, there exists no causal path in the provenance
graph $G$ from a node with trust level below threshold $\theta$ to an
\textsc{Allow} verdict on a tool call that does not traverse at least one
enforcement check.
\end{theorem}

\begin{proof}
We prove by construction, showing that every possible path from an untrusted
source to an \textsc{Allow} verdict must traverse the enforcement pipeline.

\textbf{Step 1: Complete mediation.}  \armname{} operates as an MCP proxy.
Every \texttt{tool\_use} message from the LLM client passes through
\texttt{GraphAwareEngine.evaluate()} before reaching the upstream tool server.
There is no code path from the client to the server that bypasses the engine.
Therefore, every tool call $c \in V_{\textsc{Call}}$ has a corresponding
evaluation.

\textbf{Step 2: Layer evaluation is total.}  For each tool call, the
\texttt{LayeredPolicyEngine} evaluates layers L1, L2G, L3, and L4 in order.
Evaluation terminates on the first \textsc{Deny} or after all layers return
\textsc{Pass}.  The engine returns exactly one of $\{\textsc{Allow},
\textsc{Deny}\}$.

\textbf{Step 3: L2G enforces provenance.}  Layer~L2G computes
$\textsc{MinTrust}(c)$ by traversing all ancestors of $c$ in $G$.  If
$\textsc{MinTrust}(c) < \theta$, the layer returns \textsc{Deny}.
Additionally, if $\textsc{CounterfactualChains}(c) \neq \emptyset$, the layer
returns \textsc{Deny}.

By \Cref{def:min-trust}, $\textsc{MinTrust}(c) < \theta$ if and only if there
exists an ancestor data node $d$ with $\tau(d) < \theta$---i.e., an untrusted
source is reachable from $c$.  Therefore, any path from an untrusted source to
$c$ that reaches L2G evaluation will be detected and denied.

\textbf{Step 4: L1 provides unconditional boundaries.}  Even if a path bypasses
L2G's provenance check (e.g., no \textsc{InputTo} edges were registered for the
call), L1 independently enforces:
\begin{itemize}
  \item Payload size limits (HB-1): exfiltration of bulk data is blocked.
  \item Sensitive path blocking (HB-2): access to credential files is blocked.
  \item Credential detection (HB-4): known secret formats are blocked.
  \item Schema pinning (HB-5): dynamic capability injection is blocked.
\end{itemize}

\textbf{Step 5: Composition.}  By Steps 1--4, any causal path from an untrusted
source ($\tau < \theta$) to an \textsc{Allow} verdict must traverse the
evaluation pipeline (Step~1), where L2G will detect the untrusted ancestry
(Step~3).  If the untrusted influence bypasses data-flow edges (e.g., through
the LLM's internal reasoning), L1 provides a fallback (Step~4) that blocks the
most damaging exfiltration vectors.

For causality laundering paths (no direct data-flow edges, but causal influence
through denial signals), \textsc{Counterfactual} edges make the influence
explicit in $G$, and L2G's counterfactual chain check (Step~3) detects and
denies the downstream action.

Therefore, no path from an untrusted source to an \textsc{Allow} verdict
exists that does not traverse at least one enforcement check.
\end{proof}

\begin{corollary}[Defense in Depth]
An attacker must simultaneously evade all layers (L1, L2G, L3, L4) to achieve
an \textsc{Allow} verdict on an unauthorized tool call.  Evading any single
layer is insufficient.
\end{corollary}

\begin{proof}
Follows directly from the conjunction semantics of the layer pipeline: all
layers must return \textsc{Pass} for the overall verdict to be \textsc{Allow}.
A \textsc{Deny} from any single layer is sufficient to block the call.
\end{proof}

\begin{lemma}[Monotonic Taint Propagation]
\label{lem:monotonic-taint}
For any path $v_1 \to v_2 \to \cdots \to v_k$ in $G$:
\[
  \textsc{MinTrust}(v_k) \leq \textsc{MinTrust}(v_1)
\]
\end{lemma}

\begin{proof}
By \Cref{def:min-trust}, $\textsc{MinTrust}(v_k)$ is the minimum trust over
all ancestors of $v_k$.  Since $v_1$ is an ancestor of $v_k$, the ancestors
of $v_1$ are a subset of the ancestors of $v_k$.  Taking the minimum over a
superset cannot increase the value.
\end{proof}